\documentclass[11pt]{article}
\title{Estimating Weighted Matchings in $o(n)$ Space }

\author{
Elena Grigorescu\thanks{Department of Computer Science, Purdue University, West Lafayette, IN. 
Email: {\tt elena-g@purdue.edu}.}
\and
Morteza Monemizadeh\thanks{Rutgers University, Piscataway, NJ 08854, USA. Supported by NSF CCF 1535878, IIS 1447793 and CCF 1161151. Email: {\tt mortezam@dimacs.rutgers.edu}.}
\and
Samson Zhou\thanks{Department of Computer Science, Purdue University, West Lafayette, IN. 
Email: {\tt samsonzhou@gmail.com}.}
}

\usepackage{makeidx}
\usepackage{euscript}
\usepackage{amsmath,amssymb,amsfonts}
\usepackage{dsfont}
\usepackage{latexsym}
\usepackage{subfigure}
\usepackage{graphicx}
\usepackage{times}\usepackage[scaled=0.92]{helvet}
\usepackage{euler}
\usepackage{fancybox}
\usepackage{fullpage}
\usepackage{hyperref}
\usepackage{paralist}
\usepackage{multirow}
\usepackage{color}
\usepackage{wrapfig}
\usepackage{tikz}
\usepackage{setspace}
\usepackage{algorithm}
\usepackage[noend]{algpseudocode}
\usepackage[framemethod=tikz]{mdframed}

\newenvironment{proof}{\noindent{\bf Proof : \ }}{\hfill$\Box$\par\medskip}

\newtheorem{theorem}{Theorem}
\newtheorem{corollary}[theorem]{Corollary}
\newtheorem{lemma}[theorem]{Lemma}

\newtheorem{definition}[theorem]{Definition}

\newtheorem{observation}[theorem]{Observation}

\newenvironment{proofof}[1]{\begin{trivlist} \item {\bf Proof 
#1:~~}}
  {\qed\end{trivlist}}
\renewenvironment{proofof}[1]{\par\medskip\noindent{\bf Proof of #1: \ }}{\hfill$\Box$\par\medskip}

\newcommand{\COMMENTED}[1]{{}}

\newcommand{\REAL}{\ensuremath{\mathbb{R}}}

\newcommand{\poly}{{\mathrm{poly}}}

\newcommand{\eps}{\epsilon}
\newcommand{\etal}{\textit{et al.}}

\begin{document}
\maketitle
\begin{abstract}

We consider the problem of estimating the weight of a maximum weighted matching of a weighted graph $G(V,E)$ whose edges are revealed in a streaming fashion. 
Extending the framework from Crouch and Stubbs (APPROX 2014), we develop a reduction from the maximum weighted matching problem to the maximum cardinality matching problem that only doubles the approximation factor 
of a streaming algorithm developed for the maximum cardinality matching problem. 
Our results hold for  the insertion-only and the dynamic (i.e, insertion and deletion) edge-arrival streaming models. 
The previous best-known reduction is due to Bury and Schwiegelshohn (ESA 2015)  who develop an algorithm whose approximation guarantee scales by a polynomial factor.

As an application, we obtain improved estimators for weighted planar graphs and, more generally, for weighted  bounded-arboricity graphs, by feeding into our reduction   the recent estimators due to Esfandiari \etal\  (SODA 2015) and to Chitnis \etal\  (SODA 2016).  In particular, we obtain a $(48+\eps)$-approximation estimator for the weight of a maximum weighted matching in planar graphs.

\end{abstract}
\section{Introduction}

We study the problem of estimating the weight of a maximum weighted matching in a weighted graph $G(V, E)$ whose edges arrive in a streaming fashion.   
Computing a maximum cardinality  matching (MCM) in an unweighted graph and a maximum weighted matching (MWM) of a weighted graph are fundamental problems in computational graph theory (e.g., \cite{MV80}, \cite{G90}).

Recently, the MCM and MWM problems have attracted a lot of attention in modern big data models such as streaming (e.g., \cite{FKMSZ05,McGregor:09,McGregor:05,EsfandiariHLMO15,Ahn:Guha:McGregor:12,GM12, AnandBGS12, GuptaP13, AssadiKLY16}), online  (e.g., \cite{BirnbaumM08, KorulaMZ13, BosekLSZ15}), MapReduce (e.g., \cite{LattanziMSV11}) and sublinear-time (e.g., \cite{BGS11,NS13}) models. 

Formally, the {\em Maximum Weighted Matching} problem is defined as follows.  

\begin{definition}[Maximum Weighted Matching (MWM)]
Let $G(V,E)$ be an undirected weighted graph with edge weights $w:E \rightarrow \REAL^+$. 
A matching $M$ in $G$ is a set of pairwise non-adjacent edges; that is, no two edges share a common vertex. 
A matching $M$ is called a maximum weighted matching of graph $G$ if its weight $w({M})=\sum_{\text{edge } e \in{M}} w(e)$ is maximum. 
\end{definition}

\noindent
If the graph $G$ is unweighted (i..e, $w:E \rightarrow \{1\}$ ), 
the maximum weighted matching problem becomes the \emph{Maximum Cardinality Matching} (MCM) problem.

In streaming models, the input graph is massive and the algorithm can only use a small amount of working space to solve a computational task.  In particular, the algorithm cannot store the entire graph $G=(V, E)$ in memory, but can only operate with a sublinear amount of space, preferably $o(n)$, where $|V|=n$. However, many tasks are not solvable in this amount of space,  and  in order to deal with such problems, the semi-streaming model \cite{FKMSZ05, Muthukrishnan05} was proposed, which allows $\mathcal{O}(n\,\text{polylog}(n))$ amount of working space. Both these settings have been studied in the adversarial model, where the edge order may be worst-case, and in the random order model, where the order of the edges is a uniformly random permutation of the set of edges.

For matching problems, if the goal is to output a set of edges that approximates the optimum matching, algorithms that maintain only $\tilde{\mathcal{O}}(n)$ edges cannot achieve better than $(e/e-1)$-approximation ratio (\cite{GoelKK12}, \cite{Kapralov13}). Showing upper bounds has drawn a lot of recent interest  (e.g., \cite{FKMSZ05}, \cite{KMM12}, \cite{McGregor:05}, \cite{Z08}, \cite{ELMS11}),  including a recent result \cite{GrigorescuMZ16} showing a $3.5$-approximation, which improves upon the previous $4$-approximation of \cite{CS14}.

If, on the other hand, the goal is to output only an estimate of the size of the matching, and not a matching itself, algorithms that use only $o(n)$ space are both desirable and possible. Surprisingly, very little is known about MWM/MCM in this model.  Recent work by Kapralov \etal\  \cite{KapralovKS14} shows the first $\text{polylog}(n)$ approximate estimator using only $\text{polylog}(n)$ space for the MCM problem. Further, if $\tilde{\mathcal{O}}(n^{2/3})$ space is allowed, then constant factor approximation algorithms are possible \cite{EsfandiariHLMO15}.

In a recent work,  Bury and Schwiegelshohn \cite{BuryS15}  consider the MWM problem in $o(n)$ space, showing a reduction to the  MCM problem, that scales the approximation factor polynomially. In particular, they are the first to show a constant factor estimator for weighted graphs  with bounded arboricity. Their results hold in the adversarial insertion-only model (where the updates are only edge insertion), and in the dynamic models (where the updates are both edge insertion and deletion). 
They also provide an $\Omega(n^{1-\eps})$ space  lower bound  to estimate the matching within $1+\mathcal{O}(\eps)$. Our results significantly improve the current best-known upper bounds of \cite{BuryS15}, as detailed in the next section.

\section{Our Contribution}

We extend the framework of \cite{CS14} to show a reduction from MWM to MCM that preserves the approximation within a factor of $2(1+\eps)$. 
Specifically,  given a $\lambda$-approximation estimation for the size of a maximum cardinality matching,  
the reduction provides a $(2(1+\epsilon)\cdot \lambda)$-approximation estimation of the weight of a maximum weighted matching.  Our algorithm works both in the insertion-only streaming model, and in the dynamic setting. In both these models the edges appear in adversarial order.
 
We next state our main theorem. As it is typical for sublinear space algorithms, we assume that  the edge-weights of $G=(V, E)$ are bounded by  $\poly(n)$.
  
\begin{center}
\shadowbox{
\parbox{0.82\columnwidth} {
\begin{theorem}\label{thm:main}
Suppose there exists a streaming algorithm (in insertion-only, or dynamic streaming model) 
that estimates the size of a maximum cardinality matching of an unweighted graph within a factor of $\lambda$, 
with probability at least $(1-\delta)$, using $S(n,\delta)$ space. Then, for every $\epsilon>0$, 
there exists a streaming algorithm that estimates the weight of a maximum weighted matching of a weighted graph 
within a factor of $2\lambda(1+\epsilon)$, with probability at least $(1-\delta)$, using $\mathcal{O}\left(S\left(n,\frac{\delta}{c\log n}\right) \log n\right)$ space.
\end{theorem}
}}
\end{center}

\noindent
We remark that if the estimator for MCM is specific to a monotone graph property 
(a property of graphs that is closed under edge removal), then our algorithm can use it as a subroutine 
to obtain an estimator for MWM in the weighted versions of the graphs with such properties 
(instead of using a subroutine for general graphs, which may require more space, or provide worse approximation guarantees).  
 
Our result improves the result of \cite{BuryS15}, who show a reduction from MWM to MCM that achieves 
a $\mathcal{O}(\lambda^4)$-approximation estimator for MWM, given a $\lambda$-approximation estimator for MCM.
Their reduction also allows extending MCM estimators to MWM estimators  in monotone graph properties.  

In particular, using specialized estimators for graphs of bounded arboricity, we obtain improved approximation 
guarantees compared with the previous best results of \cite{BuryS15}, as explained in Section \ref{sec:applications}, e.g., Table \ref{table:results}. In addition, our algorithm is natural and allows for a clean analysis.

\subsection{Applications} \label{sec:applications}
 Theorem \ref{thm:main} has immediate consequences  for computing MWM in graphs with bounded arboricity. A graph $G=(V,E)$ has arboricity $\nu$ if
\[\nu=\max_{U\subseteq V}\left\lceil\frac{|E(U)|}{|U|-1}\right\rceil,\]
where $E(U)$ is the subset of edges with both endpoints in $U$. The class of graphs with 
bounded arboricity includes several important families of graphs, such as planar graphs, or more generally, graphs with bounded degree, 
genus, or treewidth. Note that these families of graphs are monotone.

\paragraph{Graphs with Bounded Arboricity in the Insert-only Model}Esfandiari \etal\  \cite{EsfandiariHLMO15} provide 
an estimator for the size of a maximum cardinality matching of an unweighted graph in the insertion-only streaming model 
(for completeness we state their result as Theorem \ref{thm:arboricity:MCM:insertion}  in the Preliminaries). 
Theorem \ref{thm:main}, together with Theorem \ref{thm:arboricity:MCM:insertion} (due to \cite{EsfandiariHLMO15}) implies the following result.

\begin{theorem}\label{thm:arboricity}
Let $G$ be a weighted graph with arboricity $\nu$ and $n=\omega(\nu^2)$ vertices. 
Let $\epsilon,\delta\in(0,1)$. 
Then, there exists an algorithm that estimates the weight of a MWM in $G$ within a $2 \lambda$-approximation factor, 
where $\lambda=(5\nu+9)(1+\epsilon)$,  in the insertion-only streaming model, 
with probability at least $(1-\delta)$, using $\tilde{\mathcal{O}}(\nu\epsilon^{-2}\log(\delta^{-1})n^{2/3})$\footnote{$\tilde{\mathcal{O}}(f)=\tilde{\mathcal{O}}(f\cdot (\log n)^c)$ for a large enough constant $c$.} space.
Both the update time and final processing time are $\mathcal{O}(\log(\delta^{-1})\log n)$.
\end{theorem}

\noindent
In particular, for planar graphs, $\nu=3$ and by choosing $\delta=n^{-1}$  in Theorem \ref{thm:arboricity}, 
and $\epsilon$ as a small constant, the output of our algorithm is within $(48+\epsilon)$-approximation factor 
of a MWM, with probability at least $1-\frac{1}{n}$, using $\tilde{\mathcal{O}}(n^{2/3})$ space. 
The previous result of \cite{BuryS15} gave an approximation factor of $>3\cdot 10^6$ for planar graphs. 
\vskip 0.2in\noindent
Table \ref{table:results} summarizes the state of the art for MWM.
\begin{table}[htb]
\label{table:results}
\begin{center}
\begin{tabular}{c|c|c}
& Approximation for Planar Graphs & Approximation for Graphs with Arboricity $\nu$ \\\hline
\cite{BuryS15} & $>3\cdot 10^6$ & $12(5\nu+9)^4$ \\\hline
Here & $48+\epsilon$ & $2(5\nu+9)+\epsilon$ \\\hline
\end{tabular}
\vskip 0.1in\noindent
Table \ref{table:results}: The insertion-only streaming model requires $\tilde{\mathcal{O}}(\nu\epsilon^{-2}\log(\delta^{-1})n^{2/3})$ space for all graph classes, 
while the dynamic streaming model requires $\tilde{\mathcal{O}}(\nu\epsilon^{-2}\log(\delta^{-1})n^{4/5})$ space for all graph classes.
\end{center}
\end{table}

\paragraph{Graphs with Bounded Arboricity in the Dynamic Model}
Our results also apply to the dynamic model. Here we make use of the recent result of Chitnis \etal\  \cite{ChitnisCEHMMV16} that provides an estimator for MCM in the dynamic model (See Theorem \ref{thm:arboricity:MCM:dynamic} in the Preliminaries).

Again, Theorem \ref{thm:arboricity:MCM:dynamic} satisfies the conditions of Theorem \ref{thm:main} with $\lambda=(5\nu+9)(1+\epsilon)$, and consequently, we have the following application. 
\begin{theorem}\label{thm:arboricity:dynamic}
Let $G$ be a weighted graph with arboricity $\nu$ and $n=\omega(\nu^2)$ vertices. Let $\epsilon,\delta\in(0,1)$. 
Then, there exists an algorithm that estimates the weight of a maximum weighted matching in $G$ within a $2(5\nu+9)(1+\epsilon)$-factor 
in the dynamic streaming model with probability at least $(1-\delta)$, using $\tilde{\mathcal{O}}(\nu\epsilon^{-2}\log(\delta^{-1})n^{4/5})$ space.
Both the update time and final processing time are $\mathcal{O}(\log(\delta^{-1})\log n)$.
\end{theorem}
\noindent
In particular, for planar graphs, $\nu=3$, and by choosing $\delta=n^{-1}$ and $\epsilon$ as a small constant, 
the output of our algorithm is a $(48+\epsilon)$-approximation of the weight of a maximum weighted matching 
with probability at least $1-\frac{1}{n}$ using at most $\tilde{\mathcal{O}}(n^{4/5})$ space.

We further  remark that if $2$-passes over the stream are allowed, 
then we may use the recent results of \cite{ChitnisCEHMMV16} to obtain a $(2(5\nu+9)(1+\epsilon))$-approximation 
algorithm for MWM using only $\tilde{\mathcal{O}}(\sqrt{n})$ space.
\subsection{Overview}

We start by splitting the input stream into $\mathcal{O}(\log n)$ substreams $S_1,S_2,\cdots$, 
such that substream $S_i$ contains every edge $e \in E$ whose weight is at least $(1+\epsilon)^i$, that is, $w(e)\ge(1+\epsilon)^i$. 
Splitting the stream into sets of edges of weight only bounded below was used by Crouch and Stubbs in \cite{CS14}, leading to better approximation algorithms for MWM in the semi-streaming model. 

The construction from \cite{CS14} explicitly saves maximal matchings in multiple substreams by employing a greedy strategy for each substream. Once the stream completes, the algorithm from \cite{CS14} \emph{again} uses a greedy strategy, by starting from the substream of highest weight and proceeding downward to streams of lower weight. In each substream, the algorithm from \cite{CS14} adds as many edges as possible, while retaining a matching. However, with $o(n)$ space, we cannot store maximal matchings in memory and so we no longer have access to an oracle that explicitly returns edges from these matchings.

Instead, for each substream ${S}_i$, we treat its edges as unweighted edges and apply a MCM \emph{estimator}. We then \emph{implicitly} apply a greedy strategy, where we iteratively add as many edges possible from the remaining substreams of highest weight, tracking an estimate for both the weight of a maximum weighted matching, and the number of edges in the corresponding matching. The details of the algorithm appear in Section \ref{sec:alg}.

In our analysis, we use the simple but critical fact that, at any point, edges in our MWM estimator can conflict with at most two edges in the MCM estimator, similar to an idea used in \cite{CS14}. Therefore, if the MCM estimator for a certain substream is greater than double the number of edges in the associated matching, we add the remaining edges to our estimator, as shown below in Figure \ref{fig1}. Note that in some cases, we may discard many edges that the algorithm of \cite{CS14} chooses to output, but without being able to keep a maximal matching, this may be unavoidable.

More formally, for each $i$, let $U^*_i$ be a maximum cardinality matching for $S_i$. 
Then each edge of $U^*_i$ intersects with either one, or two edges of $U^*_j$, for all $j<i$. 
Thus, if $|U^*_{i-1}|>2|U^*_i|$, then at least $|U^*_{i-1}|-2|U^*_i|$ edges from $U^*_{i-1}$ can be 
added to $U^*_i$ while remaining a matching. 
We use a variable $B_i$ to serve as an estimator for this lower bound on the number of edges in 
a maximum weighted matching, including edges from $U^*_j$, for $j\ge i$. 
We then use the estimator for MCM in each substream $i$ as a proxy for $U^*_i$. 

\begin{figure}[!ht]
\label{fig1}
\centering
\begin{tikzpicture}[scale=2]
\tikzstyle{vertex} = [circle, minimum width=6pt, fill, inner sep=0pt]
\draw[black] (5cm,0cm) node[vertex] (v1){}
-- (6cm,0cm) node[vertex] (v2){};
\draw[black] (5cm,0.3cm) node[vertex] (v3){}
-- (6cm,0.3cm) node[vertex] (v4){};
\draw[black] (5cm,0.6cm) node[vertex] (v5){}
-- (6cm,0.6cm) node[vertex] (v6){};

\draw[dashed] (2cm,0cm) node[vertex] (v7){}
-- (3cm,0cm) node[vertex] (v8){};
\draw[dashed] (2cm,0.2cm) node[vertex] (v9){}
-- (3cm,0.2cm) node[vertex] (v10){};
\draw[dashed] (2cm,0.4cm) node[vertex] (v11){}
-- (3cm,0.4cm) node[vertex] (v12){};
\draw[dashed] (2cm,0.6cm) node[vertex] (v13){}
-- (3cm,0.6cm) node[vertex] (v14){};

\draw[dashed] (v1){}
-- (4.13cm,0.5cm) node[vertex] (a1){};
\draw[dashed] (v2){}
-- (6.87cm,0.5cm) node[vertex] (a2){};
\draw[dashed] (v3){}
-- (4.13cm,0.8cm) node[vertex] (a3){};
\draw[dashed] (v4){}
-- (6.87cm,0.8cm) node[vertex] (a4){};
\draw[dashed] (v5){}
-- (4.13cm,1.1cm) node[vertex] (a5){};
\draw[dashed] (v6){}
-- (6.87cm,1.1cm) node[vertex] (a6){};

\draw (2.5cm,0.5cm) ellipse (1cm and 1cm);
\node at (2.5cm,1cm)%
			{$U^*_{i-1}$};
\filldraw[shading=radial,inner color=white, outer color=gray!75, opacity=0.2] (4.5cm,0.5cm) ellipse (0.7cm and 1.2cm); 
\node at (4.23cm,0cm)%
			{$U^*_{i-1}$};
\filldraw[shading=radial,inner color=white, outer color=gray!75, opacity=0.2] (6.5cm,0.5cm) ellipse (0.7cm and 1.2cm); 
\node at (6.77cm,0cm)%
			{$U^*_{i-1}$};

\draw (5.5cm,0.5cm) ellipse (0.7cm and 1.2cm);
\node at (5.5cm,1cm)%
			{$U^*_i$};
\end{tikzpicture}\\
Figure \ref{fig1}: If $|U^*_i|>2|U^*_{i-1}|$, then some edge(s) from $U^*_{i-1}$ can be added while maintaining a matching.
\end{figure}

Our algorithm differs from the algorithm of \cite{BuryS15} in several points.  They consider substreams $S_i$ containing the edges with weight $[2^i,2^{i+1})$, and their algorithm estimates the number of each edges in each stream, and chooses to include the edges if both the number of the edges and their combined weight exceed certain thresholds, deemed to contribute a significant value to the estimate. However, this approach may not capture a small number of edges which nonetheless contribute a significant weight. 

Our greedy approach is able to handle both these facets of a MWM problem. Namely, by greedily  taking as many edges as possible from the heavier substreams, and then accounting for edges that may be conflicting with these in the next smaller substream, we are able to account for most of the weight. The formal analysis appears in Section \ref{sec:analysis}.
\section{Preliminaries}

Let $S$ be a stream of insertions of edges of an underlying undirected weighted graph
$G(V,E)$ with weights $w:E \rightarrow \REAL$. We assume that vertex set $V$ is fixed and given, and the size of $V$
is $|V|=n$. Observe that the size of stream $S$ is $|S|\le\binom{n}{2} = \frac{n(n-1)}{2}\le n^2$, so that we may assume that
$\mathcal{O}(\log|S|)=\mathcal{O}(\log n)$. Without loss of generality we assume that 
at time $i$ of stream $S$, edge $e_i$ arrives (or is revealed). 
Let $E_i$ denote those edges which are inserted (revealed) up to time $i$, 
i.e., $E_i=\{e_1,e_2,e_3,\cdots,e_i\}$. 
Observe that at every time $i\in [|S|]$ we have $|E_i|\le\binom{n}{2}\le n^2$, 
where $[x]=\{1,2,3,\cdots,x\}$ for some natural number $x$. 
We assume that at the end of stream $S$ all edges of graph $G(V,E)$ arrived, that is, $E=E_{|S|}$. 

We assume that there is a unique numbering for the vertices in
$V$ so that we can treat $v \in V$ as a unique number $v$ for $1 \le v \le n=|V|$.
We denote an undirected edge in $E$ with two endpoints $u,v\in V$ by
$(u,v)$. The graph $G$ can
have at most $\binom{n}{2} = n(n-1)/2$ edges.
Thus, each edge can also be thought of as referring to a unique number
between 1 and $\binom{n}{2}$.

The next theorems imply our results for graphs with bounded arboricity in the insert-only and dynamic models.
\begin{theorem}\label{thm:arboricity:MCM:insertion}\cite{EsfandiariHLMO15}
Let $G$ be an unweighted graph with arboricity $\nu$ and $n=\omega(\nu^2)$ vertices. Let $\epsilon,\delta\in(0,1)$ be two arbitrary positive values less than one. 
There exists an algorithm that estimates the size of a maximum matching in $G$ within a $(5\nu+9)(1+\epsilon)$-factor in the insertion-only streaming model with probability at least $(1-\delta)$, using $\tilde{\mathcal{O}}(\nu\epsilon^{-2}\log(\delta^{-1})n^{2/3})$ space.
Both the update time and final processing time are $\mathcal{O}(\log(\delta^{-1}))$.
In particular, for planar graphs, we can $(24+\epsilon)$-approximate the size of a maximum matching with probability at least $1-\delta$ using $\tilde{\mathcal{O}}(n^{2/3})$ space.
\end{theorem}

\begin{theorem}\label{thm:arboricity:MCM:dynamic}\cite{ChitnisCEHMMV16}
Let $G$ be an unweighted graph with arboricity $\nu$ and $n=\omega(\nu^2)$ vertices. Let $\epsilon,\delta\in(0,1)$ be two arbitrary positive values less than one. 
There exists an algorithm that estimates the size of a maximum matching in $G$ within a $(5\nu+9)(1+\epsilon)$-factor in the dynamic streaming model 
with probability at least $(1-\delta)$, using $\tilde{\mathcal{O}}(\nu\epsilon^{-2}\log(\delta^{-1})n^{4/5})$ space.
Both the update time and final processing time are $\mathcal{O}(\log(\delta^{-1}))$.
In particular, for planar graphs, we can $(24+\epsilon)$-approximate the size of a maximum matching with probability at least $1-\delta$ using $\tilde{\mathcal{O}}(n^{4/5})$ space.
\end{theorem}
\section{Algorithm}
\label{sec:alg}
For a weighted graph $G(V,E)$ with weights $w:E\rightarrow\REAL$ such that the minimum weight of an edge is at least $1$ and the maximum weight $W$ of an edge is polynomially bounded in $n$, i.e., 
$W=n^c$ for some constant $c$, for $T=\lceil \log_{1+\epsilon} W\rceil $, we create $T+1$ substreams such that substream $S_i=\left\{e\in S: w(e)\ge(1+\epsilon)^i\right\}$. 

Given access to a streaming algorithm {\sf MCM Estimator} which estimates the size of a maximum cardinality matching of an unweighted graph $G$ 
within a factor of $\lambda$ with probability at least $(1-\delta)$, we use {\sf MCM Estimator} as a black box algorithm on each $S_i$ and record the estimates. 
 In general, for a substream $S_i$, we track an estimate $A_i$, of the weight of a maximum weighted matching of the subgraph whose edges are in the substream $S_i$, along with an estimate, 
 $B_i$, which represents the number of edges in our estimate $A_i$. 
 The estimator $B_i$ also serves as a running lower bound estimator for the number of edges in a maximum matching. 
 We greedily add edges to our estimation of the weight of a maximum weighted matching of graph $G$. 
 Therefore, if the estimator $\widehat{M_{i-1}}$ for the maximum cardinality matching of the substream $S_{i-1}$ is more than double the number of 
 edges in $B_i$ represented by our estimate $A_i$ of the substream $S_i$, we let $B_{i-1}$ be $B_i$ plus the difference $\widehat{M_{i-1}} -2B_i$, and 
 let $A_{i-1}$ be $A_i$ plus $(\widehat{M_{i-1}} -2B_i) \cdot (1+\epsilon)^{i-1}$. 
We iterate through the substream estimators, starting from the substream $S_T$ of largest weight, and proceeding downward to substreams of lower weight. 
We initialize our greedy approach by setting $B_T=\widehat{M_T}$, equivalent to taking all edges in $\widehat{M_T}$.

\begin{algorithm*}[hbt]
\caption{Estimating Weighted Matching in Data Streams}
\textbf{Input:} A stream $S$ of edges of an underlying graph $G(V,E)$ with weights $w:E \rightarrow\REAL^+$ such that the maximum weight $W$ 
of an edge is polynomially bounded in $n$, i.e, $W=n^c$ for some constant $c$.\\
\textbf{Output:} An estimator $\hat{A}$ of $w(M^*)$, the weight of a maximum weighted matching $M^*$, in $G$. 
\begin{algorithmic}[1]
\State{Let $A_i$ be a running estimate for the weight of a maximum weighted matching.}
\State{Let $B_i$ be a running lower bound estimate for the number of edges in a maximum weighted matching.}
\State{Initialize $A_{T+1}=0$, $B_{T+1}=0$, and $\widehat{M_{T+1}}=0$. }
\For{$i=T$ to $i=0$}
\State{Let $S_i=\{e\in S: w(e)\ge(1+\epsilon)^i\}$ be a substream of $S$  of edges whose weights are at least $(1+\epsilon)^i$.}
\State{Let $S'_i$ be unweighted versions of edges in $S_i$.}
\State{Let $\widehat{S'_i}$ be the output of {\sf MCM Estimator} for each $S'_i$ with parameter $\delta'=\frac{\delta}{T}$.}
\State{Let $\widehat{M_i}=\max(\widehat{M_{i+1}},\widehat{S'_i})$.}
\State{Set $\Delta_i=\max(0,\lceil\widehat{M_i}-2B_{i+1}\rceil)$.}
\State{Update $B_i=B_{i+1}+\Delta_i$.}
\State{Update $A_i=A_{i+1}+(1+\epsilon)^i\Delta_i$.}
\EndFor
\State{Output estimate $\hat{A}=A_0$.}
\end{algorithmic}
\end{algorithm*}

\noindent
We note that the quantities $A_i$ and $B_i$ satisfy the following properties, which will be useful in the analysis.

\begin{observation}
\label{obs:a:sum}
$A_j=\sum_{i=j}^{T}(1+\epsilon)^i\Delta_i$
\end{observation}

\begin{observation}
\label{obs:b:sum}
$B_j=\sum_{i=j}^{T}\Delta_i$
\end{observation}
\section{Analysis}
\label{sec:analysis}
\begin{lemma}
\label{lem:b:bounds}
For all $i$, $B_i\le\widehat{M_i}\le2B_i$.
\end{lemma}
\begin{proof}
We prove the statement by induction on $i$, starting from  $i=T$ down to $i=0$. 
For the base case $i=T$, we initialize $B_{i+1}=0$. In particular, $\Delta_i=\widehat{M_i} $, so 
$   B_i=B_{i+1}+\Delta_i=\widehat{M_i}$, and the desired inequality follows.
\vskip0.2in\noindent
Now, we suppose the claim is true for $B_{i+1}\le\widehat{M_{i+1}}\le2B_{i+1}$. 
Next, we prove it for $B_{i}\le\widehat{M_{i}}\le2B_{i}$. To prove the claim for $i$ we consider two cases. 
The first case is when $2B_{i+1}<\widehat{M_i}$. 
Then
\begingroup
\allowdisplaybreaks
\begin{align*}
B_i&=B_{i+1}+\Delta_i&\text{(By definition)}\\
&=B_{i+1}+\widehat{M_i}-2B_{i+1}&(\Delta_i=\widehat{M_i}-2B_{i+1})\\
&=\widehat{M_i}-B_{i+1}\\
&\le\widehat{M_i}
\end{align*}
\endgroup
\noindent
Additionally,
\begingroup
\allowdisplaybreaks
\begin{align*}
\widehat{M_i}&<\widehat{M_i}+(\widehat{M_i}-2B_{i+1})&(2B_{i+1}<\widehat{M_i})\\
&=2(B_{i+1}+(\widehat{M_i}-2B_{i+1}))\\
&=2(B_{i+1}+\Delta_i)&(\Delta_i=\widehat{M_i}-2B_{i+1})\\
&=2B_i&\text{(By definition)}
\end{align*}
\endgroup
and so $B_i\le\widehat{M_i}\le2B_i$.
\vskip 0.2in\noindent
The second case is when $\widehat{M_i}\le2B_{i+1}$. Then, by definition, $B_i=B_{i+1}$. Since $S'_{i+1}$ is a subset of $S'_i$, then
\begingroup
\allowdisplaybreaks
\begin{align*}
B_i=B_{i+1}&\le\widehat{M_{i+1}}&( \text{Inductive hypothesis})\\
&\le\widehat{M_i}&(\widehat{M_i}=\max(\widehat{M_{i+1}},\widehat{S'_i}))\\
&\le 2B_{i+1} = 2B_i &\text{($\widehat{M_i}\le2B_{i+1}$)}
\end{align*}
\endgroup
and again $B_i\le\widehat{M_i}\le2B_i$, which completes the proof.
\end{proof}

\begin{corollary}
\label{cor:b:bounds}
Suppose for all $i$, the estimator $\widehat{M_i}$ satisfies $\widehat{M_i}\le|U^*_i|\le\lambda\widehat{M_i}$, 
where $U^*_i$ is the size of a maximum cardinality matching of $S'_i$. Then $B_i\le|U^*_i|\le2\lambda B_i$.
\end{corollary}

\begin{proof}
By Lemma \ref{lem:b:bounds}, $\widehat{M_i}\le2B_i$, so then $\lambda\widehat{M_i}\le2\lambda B_i$. 
Similarly, by Lemma \ref{lem:b:bounds}, $B_i\le\widehat{M_i}$. But by assumption, $\widehat{M_i}\le|U^*_i|\le\lambda\widehat{M_i}$, and so
\[B_i\le\widehat{M_i}\le|U^*_i|\le\lambda\widehat{M_i}\le2\lambda B_i.\]
\end{proof}

\begin{lemma}
\label{lem:partition:inequality}
Suppose for all $i$, the estimator $\widehat{M_i}$ satisfies $\widehat{M_i}\le|U^*_i|\le\lambda\widehat{M_i}$, 
where $U^*_i$ is the size of a maximum cardinality matching of $S'_i$. Then for all $j$,
\[\sum_{i=j}^T\Delta_i\le\sum_{i=j}^T|M^*\cap(S_j-S_{j+1})|\le\sum_{i=j}^T 2\lambda\Delta_i,\]
where $M^*$ is a maximum weighted matching.
\end{lemma}

\begin{proof}
Since $M^*$ is a matching, then the number of edges in $M^*$ with weight at least $(1+\epsilon)^j$ is at most $|U^*_j|$. Thus,
\[\sum_{i=j}^T|M^*\cap(S_j-S_{j+1})|\le|U^*_j|.\]
Note that by Observation \ref{obs:b:sum}, $\sum_{i=j}^T\Delta_i=B_j$, so then by Corollary \ref{cor:b:bounds},
\[\sum_{i=j}^T|M^*\cap(S_j-S_{j+1})|\le2\lambda\sum_{i=j}^T\Delta_i.\]
On the other hand, $B_i$ is a running estimate of the lower bound on the number of edges in $M^*\cap S_i$, 
 so
\[\sum_{i=j}^T\Delta_i=B_j\le\sum_{i=j}^T|M^*\cap(S_j-S_{j+1})|,\]
as desired.
\end{proof}

\begin{lemma}
\label{lem:union:bound}
With probability at least $1-\delta$, the estimator $\widehat{M_i}$ satisfies $\widehat{M_i}\le|U^*_i|\le\lambda\widehat{M_i}$ for all $i$, where $U^*_i$ is the maximum cardinality matching of $S'_i$.
\end{lemma}
\begin{proof}
Since $\widehat{M_i}\le|U^*_i|\le\lambda\widehat{M_i}$ succeeds with probability at least $1-\frac{\delta}{T}$, then the probability $\hat{M_i}$ succeeds for $i=1, 2,\ldots,T$ is at least $1-\delta$ by a union bound.
\end{proof}

\noindent
We now prove our main theorem.

\begin{proofof}{Theorem \ref{thm:main}}
We complete the proof of Theorem \ref{thm:main} by considering the edges in a maximum weighted matching $M^*$. We partition these edges by weight and bound the number of edges in each partition. We will show that $A_0 \leq w(M^*)\leq 2\lambda(1+\epsilon)A_0$. First, we have 
\begingroup
\allowdisplaybreaks
\begin{align*}
w(M^*)&=\sum_{e\in M^*}w(e)\\
&=\sum_{i=0}^T\sum_{e\in M^*\cap(S_i-S_{i+1})}w(e)&\text{(2)}\\
&\le\sum_{i=0}^T\sum_{e\in M^*\cap(S_i-S_{i+1})}(1+\epsilon)^{i+1}&\text{(3)}\\
&\le\sum_{i=0}^T|M^*\cap(S_i-S_{i+1})|(1+\epsilon)^{i+1}&\text{(4)}\\
&\le\sum_{i=0}^T2\lambda\Delta_i(1+\epsilon)^{i+1}&\text{(5)}\\
&\le2\lambda(1+\epsilon)\sum_{i=0}^T\Delta_i(1+\epsilon)^i=2\lambda(1+\epsilon)A_0,&\text{(6)}\\
\end{align*}
\endgroup
where the identity in line (2) results from partitioning the edges by weight, so that $e\in M^*$ appears in $S_i-S_{i+1}$ if $(1+\epsilon)^i\le w(e)<(1+\epsilon)^{i+1}$. The inequality in line (3) results from each edge $e$ in $S_i-S_{i+1}$ having weight less than $(1+\epsilon)^{i+1}$, so an upper bound on the sum of the weights of edges in $M^*\cap(S_i-S_{i+1})$ is $(1+\epsilon)^{i+1}$ times the number of edges in $|M^*\cap(S_i-S_{i+1})|$, as shown in line (4). By Lemma \ref{lem:partition:inequality}, the partial sums of $2\lambda\Delta_i$ dominates the partial sums of $|M^*\cap(S_i-S_{i+1}|$, resulting in the inequality in line (5). The final identity in line (6) results from Observation \ref{obs:a:sum}.
Similarly,
\begingroup
\allowdisplaybreaks
\begin{align*}
w(M^*)&=\sum_{e\in M^*}w(e)\\
&=\sum_{i=0}^T\sum_{e\in M^*\cap(S_i-S_{i+1})}w(e)&\text{(2)}\\
&\ge\sum_{i=0}^T \sum_{e\in M^*\cap(S_i-S_{i+1})}(1+\epsilon)^i&\text{(3)}\\
&\ge\sum_{i=0}^T |M^*\cap(S_i-S_{i+1})|(1+\epsilon)^i&\text{(4)}\\
&\ge\sum_{i=0}^T\Delta_i(1+\epsilon)^i&\text{(5)}\\
&\ge\sum_{i=0}^T A_i=A_0,&\text{(6)}\\
\end{align*}
\endgroup
where the identity in line (2) again results from partitioning the edges by weight, so that $e\in M^*$ appears in $S_i-S_{i+1}$ if $(1+\epsilon)^i\le w(e)<(1+\epsilon)^{i+1}$. The inequality in line (3) results from each edge $e$ in $S_i-S_{i+1}$ having weight at least $(1+\epsilon)^i$, so a lower bound on the sum of the weights of edges in $M^*\cap(S_i-S_{i+1})$ is $(1+\epsilon)^i$ times the number of edges in $|M^*\cap(S_i-S_{i+1})|$, as shown in line (4). By Lemma \ref{lem:partition:inequality}, the partial sums of $|M^*\cap(S_i-S_{i+1})|$ dominates the partial sums of $\Delta_i$, resulting in the inequality in line (5). The final identity in line (6) results from Observation \ref{obs:a:sum}.
\vskip 0.2in\noindent
Thus, $\widehat{A}=A_0$ is a $2\lambda(1+\epsilon)$-approximation for $w(M^*)$.
\vskip 0.2in\noindent
Note that the assumption of Lemma \ref{lem:partition:inequality} holds with probability at least $1-\delta$ by Lemma \ref{lem:union:bound}. Since we require $\widehat{M_i}\le|U^*_i|\le\lambda\widehat{M_i}$ with probability at least $1-\frac{\delta}{T}$, then $S\left(n,\frac{\delta}{T}\right)$ space is required for each estimator. Since $T=\log W$ substreams are used and $W\le n^c$ for some constant $c$, then the overall space necessary is $S\left(n,\frac{\delta}{c\log n}\right)(c\log n)$. This completes the proof. 
\end{proofof}
\section*{Acknowledgements}
We would like to thank anonymous reviewers for their helpful comments regarding the presentation of the paper.



\newcommand{\Proc}{Proceedings of the~}
\newcommand{\STOC}{Annual ACM Symposium on Theory of Computing (STOC)}
\newcommand{\FOCS}{IEEE Symposium on Foundations of Computer Science (FOCS)}
\newcommand{\SODA}{Annual ACM-SIAM Symposium on Discrete Algorithms (SODA)}
\newcommand{\SOCG}{Annual Symposium on Computational Geometry (SoCG)}
\newcommand{\ICALP}{Annual International Colloquium on Automata, Languages and Programming (ICALP)}
\newcommand{\ESA}{Annual European Symposium on Algorithms (ESA)}
\newcommand{\CCC}{Annual IEEE Conference on Computational Complexity (CCC)}
\newcommand{\RANDOM}{International Workshop on Randomization and Approximation Techniques in Computer Science (RANDOM)}
\newcommand{\APPROX}{IInternational Workshop on Approximation Algorithms for Combinatorial Optimization Problems  (APPROX)}
\newcommand{\PODS}{ACM SIGMOD Symposium on Principles of Database Systems (PODS)}
\newcommand{\SSDBM}{ International Conference on Scientific and Statistical Database Management (SSDBM)}
\newcommand{\ALENEX}{Workshop on Algorithm Engineering and Experiments (ALENEX)}
\newcommand{\BEATCS}{Bulletin of the European Association for Theoretical Computer Science (BEATCS)}
\newcommand{\CCCG}{Canadian Conference on Computational Geometry (CCCG)}
\newcommand{\CIAC}{Italian Conference on Algorithms and Complexity (CIAC)}
\newcommand{\COCOON}{Annual International Computing Combinatorics Conference (COCOON)}
\newcommand{\COLT}{Annual Conference on Learning Theory (COLT)}
\newcommand{\COMPGEOM}{Annual ACM Symposium on Computational Geometry}
\newcommand{\DCGEOM}{Discrete \& Computational Geometry}
\newcommand{\DISC}{International Symposium on Distributed Computing (DISC)}
\newcommand{\ECCC}{Electronic Colloquium on Computational Complexity (ECCC)}
\newcommand{\FSTTCS}{Foundations of Software Technology and Theoretical Computer Science (FSTTCS)}
\newcommand{\ICCCN}{IEEE International Conference on Computer Communications and Networks (ICCCN)}
\newcommand{\ICDCS}{International Conference on Distributed Computing Systems (ICDCS)}
\newcommand{\VLDB}{ International Conference on Very Large Data Bases (VLDB)}
\newcommand{\IJCGA}{International Journal of Computational Geometry and Applications}
\newcommand{\INFOCOM}{IEEE INFOCOM}
\newcommand{\IPCO}{International Integer Programming and Combinatorial Optimization Conference (IPCO)}
\newcommand{\ISAAC}{International Symposium on Algorithms and Computation (ISAAC)}
\newcommand{\ISTCS}{Israel Symposium on Theory of Computing and Systems (ISTCS)}
\newcommand{\JACM}{Journal of the ACM}
\newcommand{\LNCS}{Lecture Notes in Computer Science}
\newcommand{\RSA}{Random Structures and Algorithms}
\newcommand{\SPAA}{Annual ACM Symposium on Parallel Algorithms and Architectures (SPAA)}
\newcommand{\STACS}{Annual Symposium on Theoretical Aspects of Computer Science (STACS)}
\newcommand{\SWAT}{Scandinavian Workshop on Algorithm Theory (SWAT)}
\newcommand{\TALG}{ACM Transactions on Algorithms}
\newcommand{\UAI}{Conference on Uncertainty in Artificial Intelligence (UAI)}
\newcommand{\WADS}{Workshop on Algorithms and Data Structures (WADS)}
\newcommand{\SICOMP}{SIAM Journal on Computing}
\newcommand{\JCSS}{Journal of Computer and System Sciences}
\newcommand{\JASIS}{Journal of the American society for information science}
\newcommand{\PMS}{ Philosophical Magazine Series}
\newcommand{\ML}{Machine Learning}
\newcommand{\DCG}{Discrete and Computational Geometry}
\newcommand{\TODS}{ACM Transactions on Database Systems (TODS)}
\newcommand{\PHREV}{Physical Review E}
\newcommand{\NATS}{National Academy of Sciences}
\newcommand{\MPHy}{Reviews of Modern Physics}
\newcommand{\NRG}{Nature Reviews : Genetics}
\newcommand{\BullAMS}{Bulletin (New Series) of the American Mathematical Society}
\newcommand{\AMSM}{The American Mathematical Monthly}
\newcommand{\JAM}{SIAM Journal on Applied Mathematics}
\newcommand{\JDM}{SIAM Journal of  Discrete Math}
\newcommand{\JASM}{Journal of the American Statistical Association}
\newcommand{\AMS}{Annals of Mathematical Statistics}
\newcommand{\JALG}{Journal of Algorithms}
\newcommand{\TIT}{IEEE Transactions on Information Theory}
\newcommand{\CM}{Contemporary Mathematics}
\newcommand{\JC}{Journal of Complexity}
\newcommand{\TSE}{IEEE Transactions on Software Engineering}
\newcommand{\TNDE}{IEEE Transactions on Knowledge and Data Engineering}
\newcommand{\JIC}{Journal Information and Computation}
\newcommand{\ToC}{Theory of Computing}
\newcommand{\Algorithmica}{Algorithmica}
\newcommand{\MST}{Mathematical Systems Theory}
\newcommand{\Com}{Combinatorica}
\newcommand{\NC}{Neural Computation}
\newcommand{\TAP}{The Annals of Probability}
\newcommand{\TCS}{Theoretical Computer Science}

\bibliographystyle{plain}
\bibliography{gem}

\end{document}